%% file: Criteria_for_Restarts_arxiv.tex
\documentclass{llncs}
\usepackage[utf8]{inputenc}
\usepackage{graphicx}
 \usepackage[caption=false]{subfig}
\usepackage{floatrow}
 
\usepackage{amsmath}  
\usepackage{amssymb}
\usepackage{amsfonts} 
 \usepackage{times} 
 \usepackage{fancyhdr}

\newcommand{\R}{\mathbb{R}}
\newcommand{\D}{\mathrm{d}}
\newcommand{\Q}{\mathfrak{Q}}
\newcommand{\erf}{\mathrm{erf}}
\newcommand{\anti}{antiderivative }

\title{Runtime Distributions and Criteria for Restarts}
\author{Jan-Hendrik Lorenz}
\institute{Institut für Theoretische Informatik, Universität Ulm, 89069 Ulm, Germany \\
\email{jan-hendrik.lorenz@uni-ulm.de}}
\begin{document}
\maketitle 

\input{quantile_arxiv}

%


\bibliographystyle{splncs03}

\input{Criteria_for_Restarts_arxiv.bbl}

\end{document}

%% file: quantile_arxiv.tex

\begin{abstract}
Randomized algorithms sometimes employ a restart strategy.
After a certain number of steps, the current computation is aborted and restarted with a new, independent random seed. 
In some cases, this results in an improved overall expected runtime. 
This work introduces properties of the underlying runtime distribution which determine whether restarts are advantageous.
The most commonly used probability distributions admit the use of a scale and a location parameter. Location parameters shift 
the density function to the right, while scale parameters affect the spread of the distribution.
It is shown that for all distributions scale parameters do not influence the usefulness of restarts and that location parameters 
only have a limited influence. This result simplifies the analysis of the usefulness of restarts.
The most important runtime probability distributions are the log-normal, the Weibull, and the Pareto distribution.
In this work, these distributions are analyzed for the usefulness of restarts.
Secondly, a condition for the optimal restart time (if it exists) is provided. The log-normal, the Weibull, 
and the generalized Pareto distribution are analyzed in this respect. Moreover, it is shown that the optimal restart time is also not 
influenced by scale parameters and that the influence of location parameters is only linear.
\end{abstract}

\section{Introduction}
Restart mechanisms are commonly used in day-to-day life. For example, when waiting for an email response, it is common to send the original email again after some time.
Therefore it is not surprising that restart strategies are used in subjects as diverse as biology (e.g. \cite{reuveni_role_2014}), physics (e.g. \cite{evans_diffusion_2011}) and computer science (e.g. \cite{Schoning_1999}, \cite{gomes1997heavy}). 
There are at least two large fields in computer science which utilize restarts. On the one hand network protocols often have a retransmission timer (e.g. TCP, see \cite{paxson2011computing}), 
after a timeout the loss of the package is assumed and therefore the message is resent. On the other hand, probabilistic algorithms often restart after a certain number of steps without finding a solution,
those algorithms are especially common for constraint satisfaction problems (CSP) and the well-known satisfiability problem (SAT).
Although in practical use, our impression is that the power of this algorithm paradigm is still underestimated.
Restarts can be used to improve the performance of an algorithm in regards to various measures. 
For example, restarting the algorithm can help to improve the completion probability when a deadline is present.
This model was studied by Wu \cite{wu2006randomization} while Lorenz \cite{lorenz2016completion} examined completion probabilities for parallel algorithms using restarts. Another measure which can benefit from restarts is the expected runtime.

An important class of distributions for which restarts often improve the expected runtime is the so-called class of heavy-tailed distributions. Heavy-tailed distributions have a tail which decays slower than an exponential. 
Crovella et al. \cite{crovella1998heavy} showed that transmission times in the world wide web follow a power-law tail which is a subclass of heavy-tailed distributions.  
Gomes et al. \cite{gomes1997heavy} observed that some instances in CSP also show a power-law tail, while
Caniou and Codognet  \cite{caniou2013sequential} examined that the log-normal distribution is a good fit for some problems in constraint satisfaction. The log-normal distribution also belongs to the 
class of heavy-tailed distributions.

Luby et al. \cite{luby_optimal_1993} introduced two of the most important restart strategies. The fixed cut-off strategy always restarts after the same (fixed) number of steps. 
They showed that this strategy is optimal for a certain (possibly infinite) number of steps.
However, in general finding the right number of steps before restarting requires extensive knowledge of the distribution.
The second strategy which Luby et al. introduced is called Luby's (universal) strategy which slowly increases the
restart times. This strategy does not require any knowledge of the underlying distribution and 
compared to the optimal strategy the expected runtime of a process utilizing Luby's strategy is only higher by a logarithmic factor. 

Since Luby's strategy does not require any a-priori knowledge of the distribution, it is nowadays more commonly used than the fixed cut-off strategy. 
One could, however, argue that better strategies are possible since, 
in a few cases, some information is available prior to the experiment. For instance, for many algorithms, the distribution for a certain class of problems has been observed empirically. 
Arbelaez et al. \cite{arbelaez2016learning} use a machine learning approach to predict the runtime distributions of several randomized algorithms.
This knowledge can be used to obtain a better restart strategy for this class of problems. Then, a speedup of, at least, 
a logarithmic factor can be expected (compared to Luby's strategy). Such a factor cannot be ignored in practice.


Choosing the wrong restart strategy can result in expected runtimes which are much worse than not restarting at all. On the other hand,
a good choice regarding the restart strategy can result in a super-exponential speedup. We believe that not enough attention has been paid
to this field of research.

\textbf{Our contribution:}
A condition for the usefulness of restart (Theorem \ref{theo:sufficient}) and a condition for optimal restart times (Theorem \ref{theo:optimal_condtion}) is obtained.
It is shown that scale parameters neither influence the usefulness nor the optimal restart times
(Theorem \ref{theo:scale} and \ref{theo:optimal_scale}). 
The log-normal and the generalized Pareto distribution are analyzed for the usefulness of restarts and their optimal restart times.
The Weibull distribution is also studied for its optimal restart times.
Parameter settings for all these distributions for which restarts are useful are obtained.
Finally, it is shown that the influence of location parameters on the usefulness of restarts is limited. This is, for all distributions discussed
here, restarts are still useful when a location parameter is present. The influence of a location Parameter on the optimal restart time
is just linear.

\section{Preliminaries}
In this section, the notation used throughout this work is introduced.
Let $X$ be a real-valued random variable. Then $F_X(t)=\textnormal{Pr}(X\leq t)$ is the cumulative distribution function (cdf) of $X$
and its derivative $f_X(t)=\frac{d}{dt} F(t)$ is the density function of $X$. 
In many cases the quantile function of the random variable $X$ is helpful, it is defined as follows:
\begin{definition}[\cite{norman1994continuous}]
 Let $X$ be a real-valued random variable with cumulative distribution function $F_X:\R \rightarrow [0,1]$.
 Then the Quantile function $Q_X: [0,1] \rightarrow \R$ is given by
 \begin{equation}
  Q_X(p)= \textnormal{inf}\{ x \mid F(x) \geq p \}.
 \end{equation}
\end{definition}
For continuous, strictly monotonically increasing cumulative distribution functions $F_X$ the quantile function $Q_X$ is the inverse function of $F_X$. 

Luby et al. defined the fixed cut-off strategy.
\begin{definition}[\cite{luby_optimal_1993}]
 Let $X$ be a random variable describing the runtime of a probabilistic algorithm $\mathcal{A}$ on some specific input. Given any value $t \in \R_+$
 a new algorithm $\mathcal{A}_t$ is obtained by restarting $\mathcal{A}$ after time $t$ has passed without finding a solution.
 Then $X_t$ is a random variable describing the runtime of $\mathcal{A}_t$.
\end{definition}
When talking about restarts, restarts using the fixed cut-off strategy are meant.
The notion of usefulness is often used in this work. Restarts are called useful if there is a $t>0$ with $E[X_t]<E[X]$.
Throughout this work, only real-valued random variables such that $F_X$, $f_X$ and $Q_X$ exist are considered.
If it is clear from the context, the subscript $X$ is omitted for the functions defined here. For the results presented here,
we assume that the number of restarts is not limited.

\section{Main Results}
\label{sec:efficient}

 Before employing a restart strategy, it should first be considered under which conditions restarts are useful at all. 
 Moorsel and Wolter \cite{moorsel_analysis_2004} obtained a condition for the usefulness of restarts: Let $T$ be a random variable describing the runtime of the process if there is a $t>0$ with $E[T]< E[T-t\mid T>t]$,
then restarts are useful. They showed that this condition is sufficient and necessary if the mean $E[T]$ exists.
This is a property which is often shown by heavy-tailed distributions. However, there are heavy-tailed distributions which do not fulfill their conditions, and there are also light-tailed distributions
which do fulfill this condition. 
In this section, another condition for the usefulness of restarts is provided and it is applied to several distributions.
 
 \subsection{Effective Restarts}
 There are several ways to describe a dataset. Two of the most commonly used values are the median and the mean, both of which describe a 'typical' value for this dataset. 
 If the mean lies to the right of the median one can speak of (positively) skewed data. While outliers contribute linearly
 to the mean, the median is very resistant to outliers. Therefore a big difference between the median and the mean can be explained by either many outliers or a few, but extreme outliers.
 In both of these cases, restarts can be an efficient way to reduce outliers. 
 Thus comparing the mean to the median yields a simple condition for the usefulness of restarts: If $Q(0.5)/E[X] < 0.5$, then restarts are beneficial. This holds because the expected runtime $E[X_{Q(p)}]$ is bounded from above by $\frac{Q(p)}{p}$. This idea can be easily generalized.
 If the mean is large because of a small number of disproportionately long runs, but there are also many short runs, then restarts are useful.
 Measuring this inequality is a well-known field in economics which is known as income inequality metrics. One of those metrics, the Lorenz curve, turns out to be helpful in the following.
 
%
%
%
%
%
%
\begin{definition}[\cite{lorenz1905methods}]
 Let $X$ be a real-valued random variable. Then the Lorenz curve $L:[0,1]\rightarrow [0,1]$ is given by:
 \begin{align}
  &L(p)= \frac{\int_0^p{Q(x)\mathrm{d}x}}{E[X]}.
\end{align}
\end{definition}

The derivative $L'$ of $L$ is given by $L'(p)=\frac{Q(p)}{E[X]}$. If the mean is infinite, then it is clear that restarts are always useful as long as some quantile exists.
This is a property which can, for example, be observed for some power-laws.
The next theorem provides a necessary and sufficient condition.
\begin{theorem}
\label{theo:sufficient}
 Let $X$ be a real-valued random variable, then restart are useful if and only if there is a $p \in [0,1)$ such that
 \begin{equation}
  \label{eq:core_condition}
  (1-p)L'(p)+L(p)<p.
 \end{equation}
\end{theorem}
\begin{proof}
 The expected runtime with an unbounded number of restarts after $Q(p)$ steps is given by (see \cite{wolter2010stochastic}):
 \begin{align}
  E[X_{Q(p)}]=\frac{1-p}{p}Q(p)+E[X \mid X < Q(p)]
 \end{align}
 The conditional expectation $E[X \mid X < Q(p)]$ is defined by $ \frac{\int_{0}^{Q(p)} xf(x) \D x}{p}$ which
 is equivalent to $\frac{\int_{0}^{p} Q(u) \D u}{p}$. This can be obtained by substitution $u=F(x)$.
 Inserting these identities into $E[X]> E[X_{Q(p)}]$ and dividing by $E[X]$ yields:
 \begin{align}
   p > (1-p)L'(p)+L(p).
 \end{align}
 This completes the proof.
 \qed
\end{proof}
The difference with the condition in \cite{moorsel_analysis_2004} is that the existence of $E[X]$ is not required. Also,
since $p$ is limited to $[0,1)$ it is algorithmically easy to find intervals for which restarts are useful, while in the condition in \cite{moorsel_analysis_2004} the variable is often unbounded.
If the condition in Equation \ref{eq:core_condition} would be an equality instead of an inequality, then the condition would describe quantiles where restarts are neither harmful nor helpful.
Wolter \cite{wolter2010stochastic} showed that for the exponential distribution restarts are neither helpful nor harmful. 

\subsection{Optimal Restarts}
\label{sec:optimal}
In the previous section, a condition for the usefulness of restarts was introduced. This section focuses on optimal restart times. 
Wolter \cite{wolter2010stochastic} provided a relationship between the optimal restart times and the inverse hazard rate.
Here a condition for the optimal restart time is shown by using quantile functions. This classification is used to analyze
the optimal restart times of several distributions.
It is shown that a condition for the optimal restart time can be expressed solely in terms of the quantile function.
\begin{theorem}
\label{theo:optimal_condtion}
 Let $X$ be a real-valued random variable with the quantile function $Q$ and its \anti $\Q$.
 Then all optimal restart times $Q(p)$ have to fulfill:
 \begin{equation}
  (p-1)Q(p)+p(1-p)Q'(p)-\Q(p)+\Q(0)=0,
 \end{equation}
 where $Q'$ is the derivative of the quantile function $Q$.
\end{theorem}
\begin{proof}
 The expected runtime under restart after $Q(p)$ steps is given by: 
 \begin{equation}
  E[X_{Q(p)}]=\frac{1-p}{p}Q(p)+\frac{\Q(p)-\Q(0)}{p}.
 \end{equation}
 By equating the derivative with zero the function can be minimized. 
 After multiplying the derivative with $p^2$ the obtained condition is:
 \begin{equation}
  (p-1)Q(p)+p(1-p)Q'(p)-\Q(p)+\Q(0)=0,
  \label{eq:condition}
 \end{equation}
where $Q'$ is the derivative of $Q$. 
This completes the proof.
\qed
\end{proof}
Next, the influence of scale parameters on the condition from Theorem \ref{theo:optimal_condtion} is investigated.

 \subsection{Scale Parameter}
 \label{sec:scale_useful}
 In this section, it is shown that for every continuous family of distributions there is one parameter which does not have an effect:
 the scale parameter.
 
 \begin{definition}[\cite{meyer1987two}]
 \label{def:scale}
  Let $X$ be a real-valued, continuous random variable with cdf $F_X$. A new random variable $Y$ with cdf $F_Y$ is 
  obtained by the identity $F_X(x)=F_Y(\frac{x}{\beta})$ for $\beta > 0$. The parameter $\beta$ is called a scale-parameter. This is
 denoted as $Y \stackrel{\text{d}}{=} \beta X$.
 \end{definition}
 With this definition, the main results of the section can be derived.
 \begin{theorem}
 \label{theo:scale}
  Let $X$ be a real-valued, continuous random variable such that restarts are useful.
  Then for $Y \stackrel{\text{d}}{=} \beta X$ restarts are also useful for every $\beta > 0$. 
 \end{theorem}
 \begin{proof}
  Let $X$ and $Y=\beta X$ be random variables with $\beta > 0$. Then the quantile function $Q_Y$, its \anti $\Q_Y$ and the mean $E[Y]$
  are given by:
  \begin{align}
   & Q_Y(p) = \beta Q_X(p), \\
   & \Q_Y(p) = \beta \Q_X(p), \\
   & E[Y] = \beta E[X]
  \end{align}
  Thus, $Y$ fulfills the condition from Theorem \ref{theo:sufficient} iff $X$ fulfills it.
  \qed
 \end{proof}

%
%
%
Since the derivative $Q'_Y$ is given by $Q'_Y=\beta Q'_X$, the result can be extended to show that scale parameters 
do not change the optimal restart time. The proof is similar in its nature and is therefore omitted.
\begin{theorem}
\label{theo:optimal_scale}
 Let $X$ be a real-valued, continuous random variable and let $q \in (0,1)$ be such that $Q_X(q)$ is the optimal restart time.
 Let $\beta >0$ be a positive, real number, then $Q_{\beta X}(q)$ is the optimal restart time for the random variable $\beta X$.
\end{theorem}
These findings show that scale parameters can be ignored in the analysis for restart times.
For several commonly used distributions the properties from Theorem \ref{theo:sufficient} and Theorem \ref{theo:optimal_condtion} can be applied.
 
 \subsection{Log-normal}
 \label{sec:log-normal-exist}
 Due to the central limit theorem, the log-normal distribution arises by the product of $n$ i.id random variables. 
 Barrero et al. \cite{barrero2015statistical} observed log-normally distributed times of several evolutionary algorithms, 
 including genetic programming, particle swarm optimization, and genetic algorithms.
 Muñoz et al. \cite{munoz2012run} empirically showed that the runtime of several path planning algorithms such as $A^*$ and $Theta^*$ 
 follow log-normal distributions. Frost et al. \cite{Frost1997CSP} argued that the runtime of several backtracking algorithms follow 
 log-normal distributions in the case of unsolvable binary CSP instances. Arbelaez \cite{arbelaez2013using} studied the runtime distributions of two SAT solvers and found
 that for randomly generated instances the log-normal distribution is a good fit. 
 Thus, the log-normal distribution is commonly used to describe the runtime of local search algorithms.
 The log-normal distribution is defined as follows:
 \begin{definition}[\cite{norman1994continuous}]
  Let $X$ be a real-valued random variable. If there are parameters $\mu >0, \sigma>0$, such that the random variable $U$ with
  \begin{equation}
    U = \frac{\log{(X)}-\mu}{\sigma}
  \end{equation}
  is standard normal distributed, then $X$ is said to be log-normally distributed. 
  In this case the mean $E[X]$ and the quantile function $Q_X$ are given by
  \begin{align}
   & E[X]=e^{\mu + \sigma^2/2}, \\
   & Q_X(p)=e^{\mu + \sigma \sqrt{2}\erf^{-1}(2p-1)},
  \end{align}
  where $\erf^{-1}$ is the inverse error function.
 \end{definition}
 The $\erf^{-1}$ function is not analytically solvable, but there are numerical approaches.
An \anti $\Q$ of $Q_X$ and the derivative $Q'$ can be obtained by calculations:
\begin{align}
 & \Q (p) =  -\frac{1}{2} e^{ \mu+\sigma^2/2}\erf\left(\frac{\sigma}{\sqrt{2}}-\erf^{-1}(2p-1) \right),\\
 & Q'(p) = e^{\mu + \sqrt{2} \cdot \sigma \cdot \erf^{-1}( 2 p-1) + \left(\erf^{-1}(2 p-1)\right)^2} \sigma \sqrt{2 \pi}.
\end{align}
Where $\erf$ is the error function which only can be computed
numerically. 
With these definitions the usefulness of restarts for the log-normal distribution can be estimated. 
\begin{theorem}
 Let $X$ be a log-normal distributed random variable. Then there is a $p\in(0,1)$ such that 
 \begin{equation}
  E[X_{Q(p)}]<E[X].
 \end{equation}
\end{theorem}
\begin{proof}
 Due to Theorem \ref{theo:scale}, the scale parameter $e^\mu$ can be ignored for the analysis of the usefulness of restarts. Therefore, let $\mu = 0$.
 Note that $E[X_{Q(p)}]$ converges to $E[X]$ as $p$ approaches one. The derivative with respect to $p$ of the expected runtime $E'[X_{Q(p)}]$ is obtained similarly to Theorem \ref{theo:optimal_condtion} and is given by:
 \begin{equation}
  E'[X_{Q(p)}]=\frac{(p-1)}{p^2}Q(p)+\frac{(1-p)}{p}Q'(p)-\frac{1}{p^2}(\Q(p)+\Q(0)).
 \end{equation}
 The limit of $E'[X_{Q(p)}]$ as $p$ approaches one is analyzed. In this case,  it can be seen that $\frac{(p-1)}{p^2}Q(p)$ converges to zero and 
 $\frac{1}{p^2}(\Q(p)+\Q(0))$ converges to $E[X]$. Therefore the analysis focuses on the limit of $(1-p)Q'(p)$. 
 \begin{align}
  & \lim_{p\rightarrow 1}(1-p)Q'(p) = \lim_{p\rightarrow 1}\frac{(1-p)}{e^{-\sqrt{2} \cdot \sigma \cdot \erf^{-1}( 2 p-1) - \left(\erf^{-1}(2 p-1)\right)^2}} \sigma \sqrt{2 \pi} \\
= & \lim_{p\rightarrow 1} \frac{\sigma \sqrt{2} e^{\sqrt{2} \cdot \sigma \cdot \erf^{-1}( 2 p-1)}}{\sqrt{2}\sigma+\erf^{-1}(2p-1)}
 = \lim_{p\rightarrow 1} \sigma^2 e^{\sigma \sqrt{2} \erf^{-1}(2p-1)} \rightarrow \infty
 \end{align}
This limit is obtained by applying L'Hospital's rule twice.
Thus, since $E[X_{Q(p)}]$ converges to $E[X]$ and $E'[X_{Q(p)}]$ approaches positive infinity, there is a $p \in (0,1)$ with $E[X_{Q(p)}]<E[X]$. 
\qed
\end{proof}

Figure \ref{fig:lognormal_region} shows parameter combinations for $p$ and $\sigma$ for which restarts are useful.
Since neither the error function nor the inverse error function can be solved analytically, numerical 
methods have been used to obtain those results. 

\begin{figure}[tbp]
\centering
\floatbox[{\capbeside\thisfloatsetup{capbesideposition={left,top},capbesidewidth=6cm}}]{figure}[\FBwidth]
{\caption{A regionplot for the log-normal distribution. 
 The blue area denotes parameter settings which fulfill the condition from Theorem \ref{theo:sufficient}.
 For $\sigma < 0.48$ the numerical approach could not find values of $p$ such that restarts are useful. 
 For high values of $\sigma$ even low values of $p$ yield an improved expected 
runtime under restart, while for low values of $\sigma$ only very high values of $p$ improve the expected runtime.}\label{fig:lognormal_region}}
 {\includegraphics[width=0.8\linewidth]{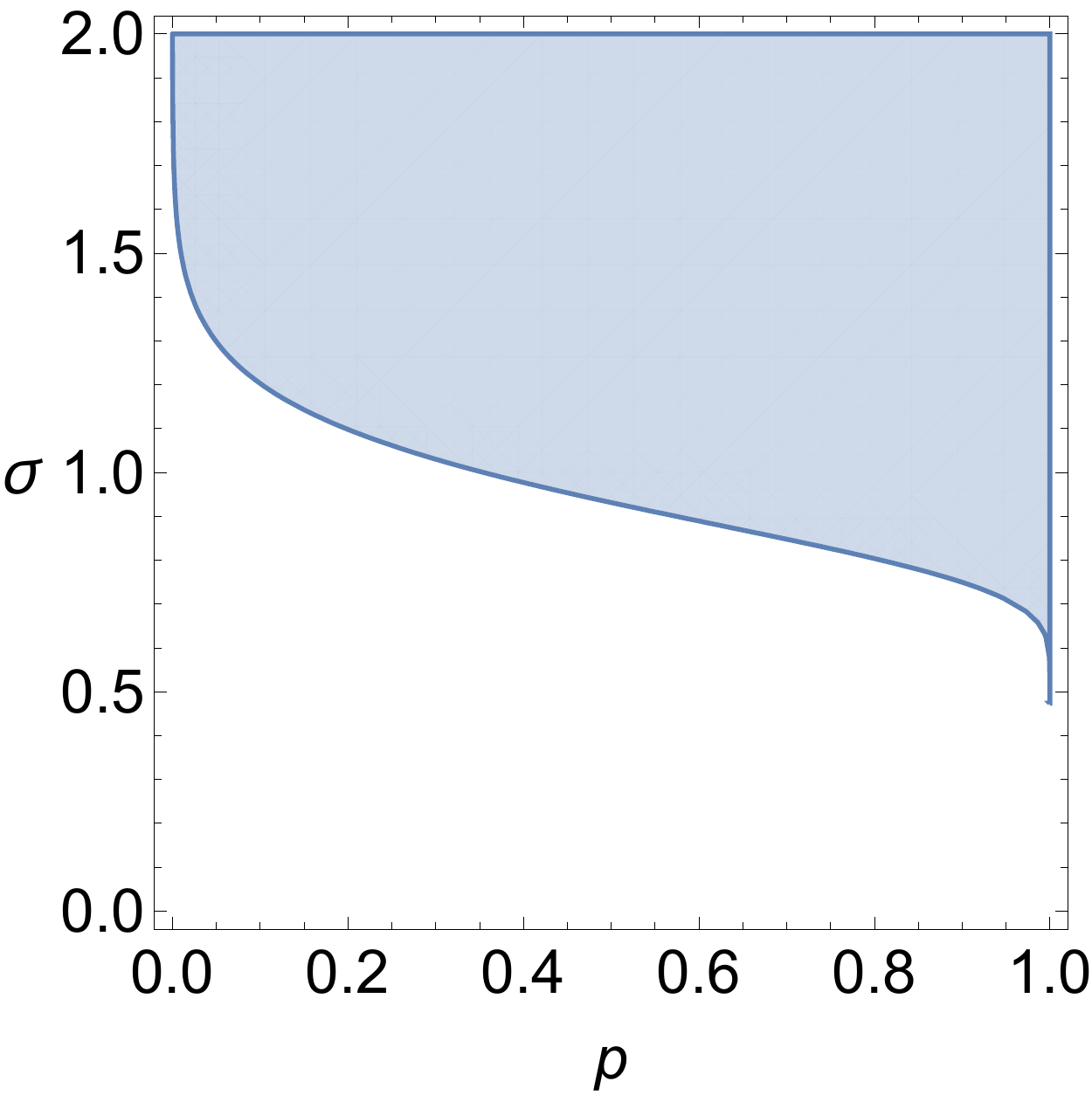}}
 \end{figure}

The optimal restart times, as presented in 
Theorem \ref{theo:optimal_condtion}, are shown in Figure \ref{fig:lognormal_optimal}.
It can be seen that for high $\sigma$ values the optimal
restart time quickly approaches $Q(0)$, while for small $\sigma$ values the optimal restart time converges to $Q(1)$.

\begin{figure}[tbp]
 \centering
\subfloat[The brown line denotes parameter settings of $\sigma$ and $p$ which minimize the expected runtime under restart.]
{
  \includegraphics[width=0.33 \textwidth]{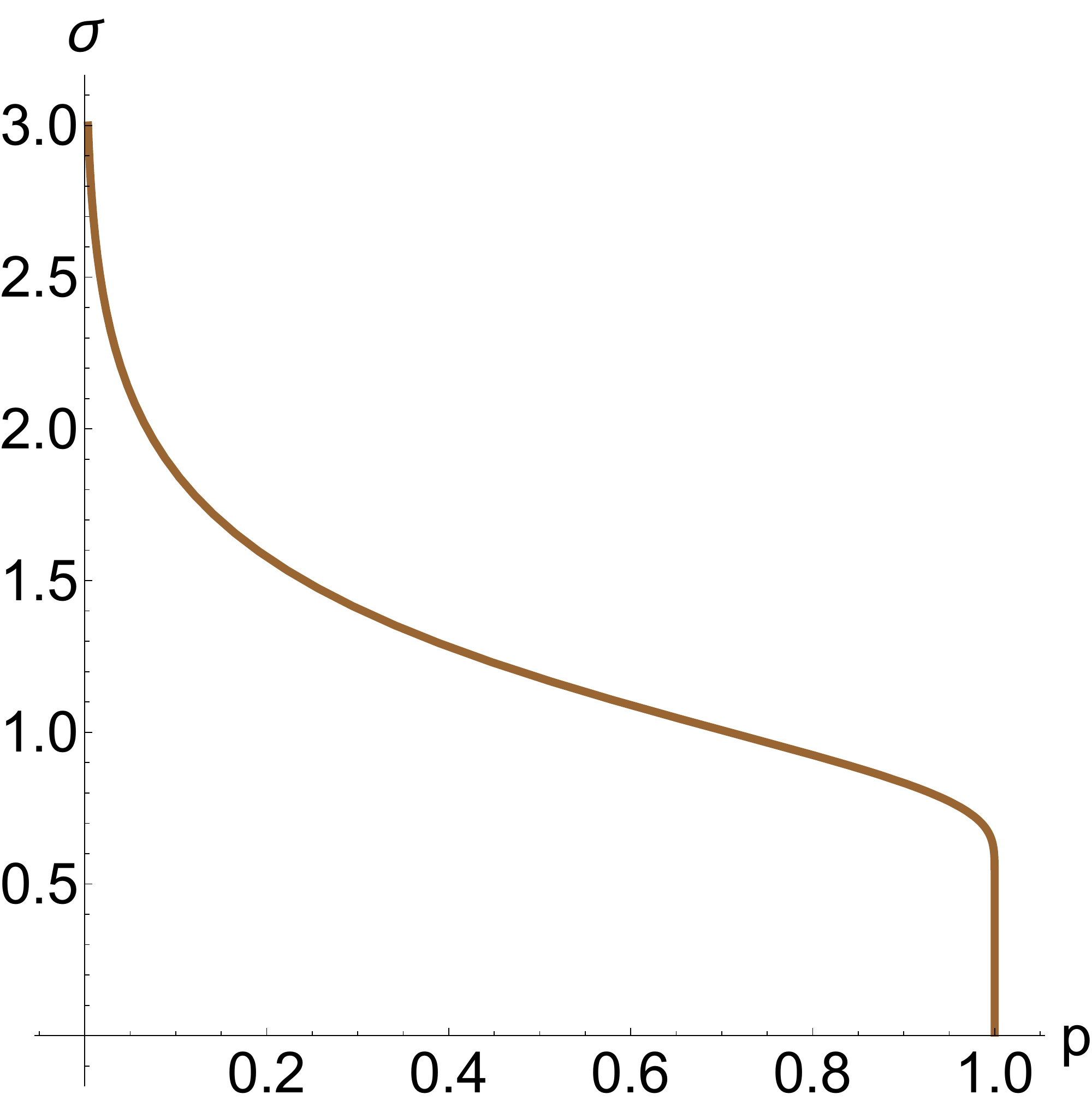}
  \label{fig:lognormal_optimal}
}
\qquad
\subfloat[This plot shows the expected runtime for a log-normal distribution without restart as a dashed line. 
 The blue line represents the expected runtime with restarts at the optimal restart time.]
 { 
    \includegraphics[width=0.58 \linewidth]{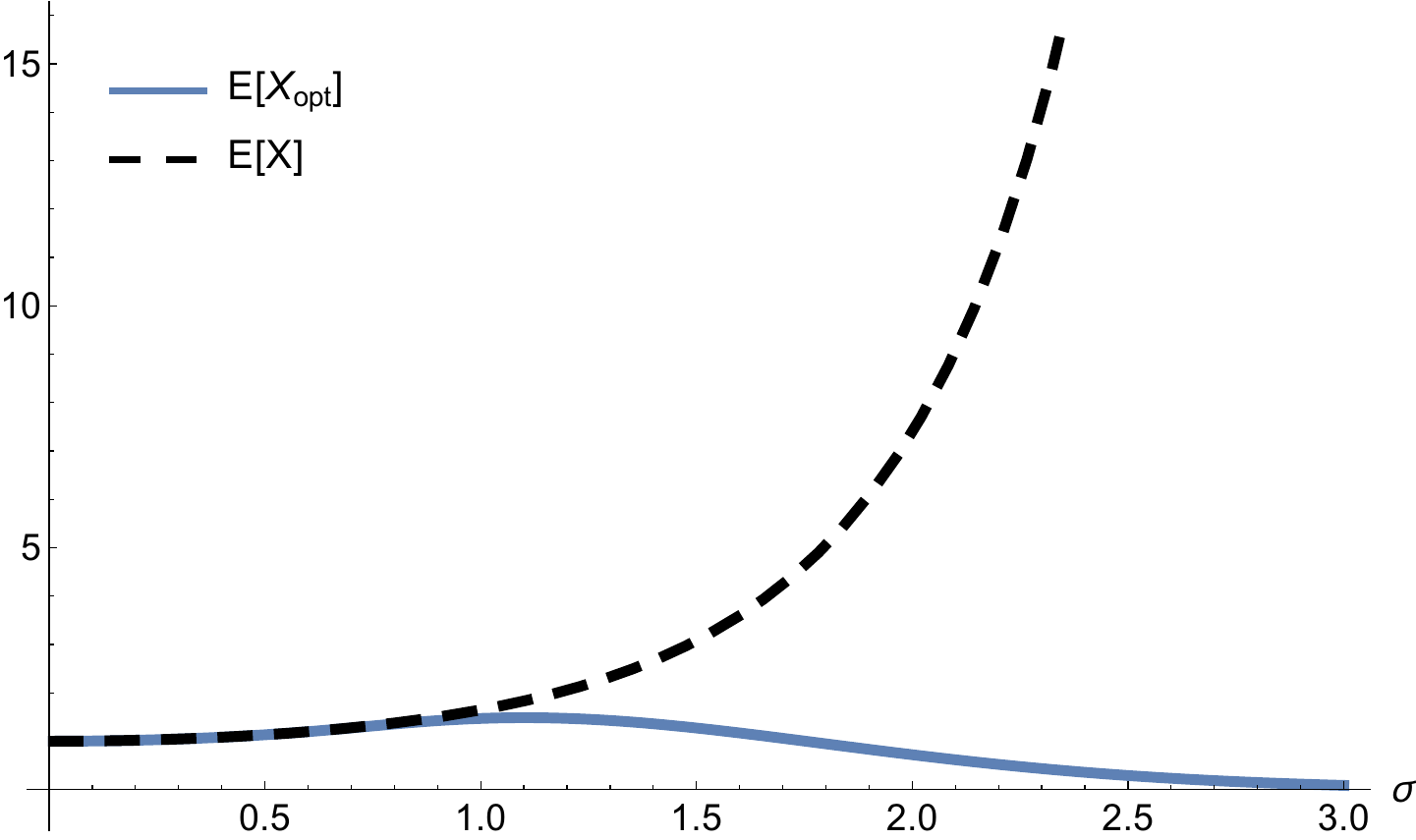}
    \label{fig:lognormal_runtime}
 }
\caption{The figures show the optimal restart quantiles for a log-normal distribution with $\mu=1$ on the left, and 
the expected runtime and the optimal mean with restarts on the right.}
\end{figure}

Figure \ref{fig:lognormal_runtime} shows the comparison of a log-normal distribution without restarts to a log-normal distribution
with restart at the optimal time.
The plot has the same shape for all values of $\mu$, and only the values on the $y$-axis differ. It can be observed 
for values up to approximately $\sigma \approx 0.8$ the difference in the expected runtime is marginal. 
On the other hand, while the expected value $e^{\mu + \frac{\sigma^2}{2}}$ behaves super-exponentially for 
increasing values of $\sigma$, the expected runtime with restarts after an optimal number of steps starts to 
decrease at about $\sigma \approx 1.1$. Therefore, in real applications restarts should only be employed for
high values of $\sigma$.

 \subsection{Generalized Pareto}
 \label{sec:pareto-exist}
 The Pickands–Balkema–de Haan theorem (see \cite{balkema1974residual}) states that the excess probability
 $\Pr(X-u \leq y \mid X > u)=\frac{F(u+y)-F(u)}{1-F(u)}$ converges in distribution towards the generalized Pareto distribution 
 for a large number of distributions as $u \rightarrow \infty$. Due to this, it can be used to model the tail of 
 distributions. The generalized Pareto distribution (GP) includes the exponential, the Pareto, and the uniform distribution.
Since the Pareto distribution is a subclass of the GP, the GP is also well suited to describe power-law decays in the tail. 
Corvella et al. \cite{crovella1998heavy} observed that network transmission times follow a power-law in the tail.
Gomes and Selman \cite{gomes1997heavy} found that the runtime of the quasi-group completion problem is also well described by a power-law.
At this point, the generalized Pareto distribution is formally defined:
\begin{definition}[\cite{norman1994continuous}]
  Let $X$ be a real-valued random variable. If the cdf $F_X$ is given by
 \begin{equation}
  F_X(x)=1-\left( 1+\frac{k x}{\sigma} \right)^{-1/k}
 \end{equation}
  for $\sigma > 0, k \in \R$, then $X$ has a generalized Pareto (GP) distribution. 
\end{definition}
For $k<1$ the mean of the GP is given by $E[X]=\frac{\sigma}{1-k}$, otherwise the mean is infinite.
For $k=0$ the GP is equivalent to the exponential distribution, and for all $k>0$ it takes
the form of a Pareto distribution. For $k\leq -0.5$ the GP has finite support, and for $k=-1$ it becomes a uniform distribution. 
Compare \cite{norman1994continuous} for these results. 
The Quantile function $Q$, an \anti $\Q$ and the derivative $Q'$ are given by:
\begin{align}
 & Q(p)= \frac{\sigma}{k} \Big( (1-p)^{-k} -1 \Big), \\
 & \Q(p)= -\frac{\sigma}{k} \Big( p+ \frac{(1-p)^{1-k}}{1-k} \Big), \\
 &  Q'(p)=\sigma (1-p)^{-1-k}.
\end{align}
These definitions can be used to analyze the usefulness of restarts.
\begin{theorem}
 Let $X$ be a generalized Pareto distribution, there is a $p \in [0,1)$ with 
 \begin{equation}
  E[X_{Q(p)}]<E[X]
 \end{equation}
 if and only if $k>0$.
\end{theorem}
\begin{proof}
Due to Theorem \ref{theo:scale}, the scale parameter $\sigma$ does not influence the usefulness of restarts.
Therefore, define $\sigma =1$ for the rest of this proof.
For the case $k\geq 1$ the mean is infinite. Thus, restarts are useful for all 
$p \in (0,1)$. Thus, only the  case where $k<1$ is considered.
Then the Lorenz curve $L$ and its derivative $L'$ are given by:
\begin{align}
 & L(p) = \frac{\Q(p)-\Q(0)}{E[X]}= \frac{1}{k}\Big( 1 - (1-p)^{1-k} -(1-k)p \Big), \\
 & L'(p)= \frac{Q(p)}{E[X]}=\frac{1-k}{k}\Big( (1-p)^{-k}-1 \Big).
\end{align}
By inserting the Lorenz function and its derivative in inequality \ref{eq:core_condition}, the following can be obtained by some calculations:
\begin{align}
 &p>-(1-p)^{1-k}+1 \\
 \Leftrightarrow & k \log{(1-p)}<0
\end{align}
For $p\in (0,1)$ the left side of the equation $k \log{(1-p)}$ is negative if and only if $k$ is positive.
Therefore restarts are useful for all $p\in [0,1)$ iff $k>0$.
 \qed
\end{proof}
 This is consistent with other results in this field. 
Wolter \cite{wolter2010stochastic} showed that Restarts are useful for the Pareto distribution,
they are not useful for the uniform distribution and are neither helpful nor harmful for the exponential distribution.
The result presented here is stronger since all of those three distributions are subclasses of the generalized Pareto distribution
while other distributions can be obtained by using different parameters. The next theorem analyzes the optimal restart time.

\begin{theorem}
 Let $X$ be a generalized Pareto distributed random variable with $k>0$, then the optimal restart time is zero and the optimal expected runtime 
 under restart $E[X_{Q(0)}]$ is $E[X_{Q(0)}]=\sigma$.
\end{theorem}
\begin{proof}
 Due to Theorem \ref{theo:optimal_scale}, the scale parameter $\sigma$ does not have any influence on the optimal restart time.
Hence, it can be set to $\sigma = 1$ for the analysis. 
Then the condition as in Theorem \ref{theo:optimal_condtion} is:
\begin{equation}
 \frac{1}{k} \left(\frac{1}{1-k}- p-\frac{(1-p)^{1-k}}{1-k} +(1-p) \left( (1-p)^{-k} - 1\right)\right)-p(1-p)^{-k}=0.
 \label{eq:pareto}
\end{equation}
By transforming the equation, this condition can be simplified to $(1-p)^{-k}(1-kp)=1$. 
This is obviously true for $p=0$. Actually, it can also be shown that $p=0$ is the only value which fulfills the condition.
Differentiating $(1-p)^{-k}(1-kp)$ with respect to $p$ yields $\frac{(k-1)kp}{(1-p)^k (p-1)}$.
For $k\in (0,1)$ this is strictly positive, and for $k \in (1, \infty)$ this is strictly negative for $p\neq 0$.
For this reason, the condition is either strictly monotonically increasing or strictly monotonically decreasing.
Therefore, every other value of $p$ does not fulfill equation \ref{eq:pareto}. 
The expected runtime is given by $E[X_{Q(p)}]=\frac{1-p}{p}Q(p)+\frac{\Q(p)-\Q(0)}{p}$. 
In case of GP the limit exists for $p\rightarrow 0$, it can be obtained by applying L'Hospital's rule and is given by
$E[X_{Q(0)}]\rightarrow Q'(0)=\sigma$.
 \qed
\end{proof}
It is noteworthy that the expected runtime does not depend on the shape parameter $k$ anymore.
When comparing expected runtimes under restart with the expected runtimes without restarts $E[X]=\frac{\sigma}{1-k}$ for $k\in (0,1)$,
it is easy to see that for $k\rightarrow 0$ both runtimes converge against $\sigma$. This is consistent with the fact that
the generalized Pareto distribution becomes the exponential distribution for $k=0$.

\subsection{Weibull}
\label{sec:weibull_optimal}
The Weibull distribution was extensively analyzed by Wolter \cite{wolter2010stochastic}. 
The Weibull distribution is one of the three limiting distributions of the Fisher-Tippett-Gnedenko Theorem in case of 
the minimum value (see for example \cite{kotz2000extreme}). Therefore it is a likely candidate when observing the minimum
of $n$ i.id random variables $X_1, \dots, X_n$. 
Frost et al. \cite{Frost1997CSP} observed that the runtime distributions of several backtracking algorithms can be reasonably well described by Weibull distributions
for solvable binary CSP instances at the $50\%$ satisfiability point.
Hoos and Stützle \cite{hoos1998evaluating} examined the runtime of a SAT-solver (GSAT). They found that for non-optimal
parameter settings the runtime can be described as a Weibull distribution.
Barrero et al. \cite{barrero2015statistical} studied generation based models without selective pressure, they
argue that the generations-to-success can be modeled by a Weibull distribution. 
The Weibull distribution is defined as follows:
\begin{definition}[\cite{norman1994continuous}]
 Let $X$ be a real-valued random variable. The random variable $X$ has a Weibull distribution if and only if
 \begin{equation}
  F_X(x)=1-e^{-(\frac{x}{a})^k}
 \end{equation}
for some fixed $k>0$. Then the quantile function $Q_X$ is given by:
\begin{equation}
 Q_X(p)=a(-\log{(1-p)})^{1/k}.
\end{equation}
\end{definition}
Wolter \cite{wolter2010stochastic} showed that
restarts are always useful for $k<1$. In case $k>1$ restarts are always harmful and in case $k=1$ the Weibull distribution becomes the 
exponential distribution, therefore restarts are neither useful nor harmful.
The same results can be obtained by the technique from Theorem \ref{theo:sufficient}. 
The derivative $Q'$ and an \anti $\Q$ of $Q$ are given as follows:
\begin{align}
 & Q'(p)= a \frac{(-\log{(1 - p)})^{\frac{1}{k}-1}}{k (1 - p)},\\
 & \Q(p) = a \gamma\left(1+\frac{1}{k}, -\log{(1-p)}\right),
\end{align}
where $\gamma(z,x)=\int_0^x t^{z-1}e^{-t}\D t$ is the incomplete gamma function. The optimal
restart time is analyzed in the next theorem.

\begin{theorem}
 Let $X$ be a Weibull distributed random variable with $k<1$. Then the optimal restart time is zero and the optimal expected runtime 
 under restart is $E[X_{Q(0)}]=0$.
\end{theorem}
\begin{proof}
	 The quantile function, its derivative and the \anti can be inserted into Theorem \ref{theo:optimal_condtion} which yields:
	 \begin{equation}
	 \gamma\left(1+\frac{1}{k}, -\log{(1-p)}\right)+(-\log{(1-p)})^{\frac{1}{k}}(1-p-\frac{p}{k}(-\log{(1-p)})^{-1})=0.
	 \label{eq:weibull_eq}
	 \end{equation}
	 The scale parameter $a$ can be ignored due to Theorem \ref{theo:optimal_scale}. Therefore, define $a=1$ for this proof.
	 Since $\frac{1}{k}-1$ is strictly positive for $k<1$ the term $\frac{p}{k}(-\log{(1-p)})^{\frac{1}{k}-1}$ approaches zero as 
	 $p$ approaches zero, hence, equation \ref{eq:weibull_eq} approaches zero. Then, the expected runtime $E[X_{Q(p)}]$ approaches
	 $Q'(p)=\frac{a(-\log{(1-p)})^{\frac{1}{k}-1}}{k(1-p)}$. 
	 \begin{equation}
	 E[X_{Q(0)}]=
	 \left\{
	 \begin{array}{ll}
	 0  & \mbox{if } k < 1 \\
	 a & \mbox{if } k = 1\\
	 \infty & \mbox{if } k>1.
	 \end{array}
	 \right.
	 \end{equation}
	 Zero is, by definition, the lowest possible runtime, restarting at $Q(0)$ is therefore the optimal restart strategy for Weibull distributions 
	 with $k<1$.
	 \qed
\end{proof}
It is remarkable that the optimal expected runtime in the case of $k<1$ is not dependent on any parameter of the distribution. 

%
%
 
 \subsection{Location Parameter}
 \label{sec:location}
Up to now, several  multiplicative variants of random variables $X$ were discussed.
Nonetheless, many commonly used distributions require an additive location parameter $b\in\mathbb{R}$ which shifts the support 
of the cdf. 
In this section, the results are augmented with this extension and it is shown that the influence of the location parameter $b$ is limited.


In the following, let $X$ be a random variable without location parameter and let $Y=X+b$ be a random variable with some location parameter $b$.
The expected value is known to be a linear function, therefore $E[Y]= E[X]+b$. Similar results follow 
easily for the quantile function. Since $F_Y(x)=F_X(x-b)$ holds $Q_Y(p)=Q_X(p)+b$ directly follows. Then an \anti of $Q_Y$ is given by $\Q_X(p)+pb$ where $\Q_X$
is an \anti of $Q_X$. Define $c\in\R$ such that $b = cE[X]$.
With these identities the usefulness of restarts can be reestimated:
\begin{align}
 &E[X]+b > \frac{1-p}{p}(Q_X(p)+b)+b+\frac{1}{p}(\Q_X(p)-\Q_X(0))\\
 \Leftrightarrow & p+(p-1)c > (1-p)L'_X(p)+L_X(p)
\end{align}
is the new condition for the usefulness of restarts.
This inequality assumes $E[X]>0$, for $E[X]<0$ the \textquoteleft greater than\textquoteright{} sign becomes
a \textquoteleft less than\textquoteright{} sign. However, this only makes sense if the location parameter shifts the support of the distribution to strictly positive values.
It can be shown that the condition
for the optimal restart time only changes by the location parameter itself. This can be shown by similar transformations; the proof is 
therefore omitted.
\begin{theorem}
 Let $X$ be a random variable and let $Y=X+b$ be a random variable with location parameter $b \in \R$.
 Then the condition for the optimal restart time for $Y$ is:
 \begin{equation}
   (p-1)Q_X(p)+p(1-p)Q'_X(p)-\Q_X(p)+\Q_X(0)=b.
   \label{eq:location_op}
 \end{equation}
\end{theorem}
If $Q'_X$ dominates $(1-p)$ in the neighborhood of $p=1$, then the left side of equation \ref{eq:location_op} approaches
infinity. This implies that restarts are useful for an arbitrarily large $b$ since $E[X_{Q(1)}]=E[X]$. 
This is the case for all distributions considered in this article.
\begin{corollary}
 Let $X$ be log-normal, GP, or Weibull distributed, with parameters such that restarts are useful. Let $b\in \R_+$, then there is a $p \in (0,1)$ with
 \begin{equation}
  E[(X+b)_{Q(p)}] < E[X+b].
 \end{equation}
\end{corollary}
Note, that this is not a general property which is true for all distributions which admit useful restarts. Counterexamples are distributions with a finite support.
With this relationship, it is reasonable to analyze distributions without location parameters and scale parameters.
This simplifies the analysis whether a restart strategy should be employed and if so, which strategy should be chosen.

\section{Discussion}
This work discussed the relationship between the quantile function and restarts using the fixed cut-off strategy. Theorem \ref{theo:sufficient} uses the quantile function
and its \anti to provide a condition for the usefulness of restarts, while Theorem \ref{theo:optimal_condtion} established a condition
for the optimal restart times.
It was proven in Theorem \ref{theo:scale} and \ref{theo:optimal_scale} that scale parameters can be ignored in the context of restarts.
In section \ref{sec:location} it was shown that the influence of location parameters on the usefulness of restarts is limited. For a large
group of distributions 
the usefulness of restarts is not affected at all by the presence of a 
location parameter. Secondly, the optimal restart times are just linearly influenced by a location parameter. 
Therefore, it often suffices to analyze the remaining parameters.\\
\indent Several commonly used distributions were observed for their usefulness under restart. 
In the following, the log-normal distribution
(compare section \ref{sec:log-normal-exist}) and the generalized Pareto distribution (compare section \ref{sec:pareto-exist}) were discussed.
It was shown that in the case of the log-normal distribution restarts are always useful.
In case of the generalized Pareto distribution, restarts are useful iff the shape parameter $k$ is greater than zero.
The optimal restart times and optimal expected runtimes under restart of the log-normal, the generalized Pareto and the Weibull distribution 
were discussed. The expected runtime without restart for the log-normal distribution increases super-exponentially in $\sigma$. For the log-normal distribution, it
was numerically observed that for increasing parameter $\sigma$ the optimal restart time is decreasing. And while the expected value without
restarts increases super-exponentially, the expected value with restart at the optimal time starts to decrease at about $\sigma \approx 1.1$.
It is also interesting to see that for low values of $\sigma$ the speedup with restarts is marginal. This is especially important 
if the parameters of the distribution are not completely known. Figure \ref{fig:log_expected_sigma07} represents an example of this behavior.
This shows that choosing a suboptimal restart time
can easily result in expected runtimes under restart which are worse than not employing a restart strategy.
Therefore, in practice it can be better to not employ a restart strategy
if the parameters are estimates and $\sigma$ is estimated to be low. 
\begin{figure}[tbp]
 \centering
 \floatbox[{\capbeside\thisfloatsetup{capbesideposition={left,top},capbesidewidth=5.8cm}}]{figure}[\FBwidth]
 {\caption{This figure depicts the expected runtime of a log-normally distributed random variable $X$ with $\mu=0$ and $\sigma=0.7$ as a dashed line.
 The blue line is the expected runtime with restart after $Q_X(p)$ steps.} \label{fig:log_expected_sigma07}}
 {\includegraphics[width=\linewidth]{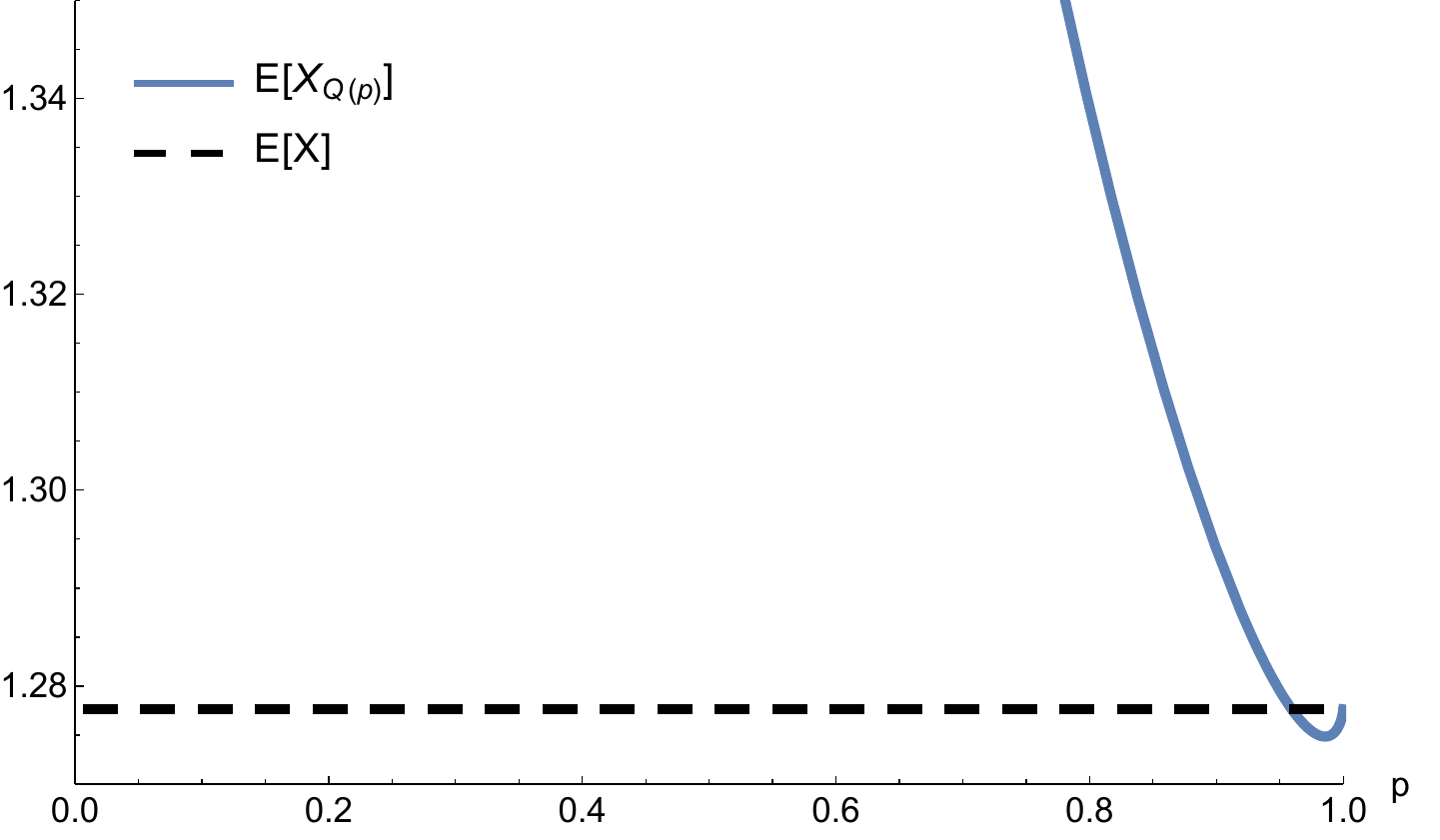}}
\end{figure}
\\
\indent
For the Weibull distribution, it was shown that the optimal restart time and 
the expected runtime is zero. A similar behavior was observed for the GP distribution.
The optimal restart time is also zero, its optimal expected runtime under restart, however, approaches $\sigma$.
It is noteworthy that in these two cases the expected runtime is no longer dependent on the shape parameter and
in the case of the Weibull distribution also not dependent from the scale parameter. 
For algorithms, an intuitive description of the restart quantile zero in, e.g., runtime distributions with location $b$ and scale $a$ emerges from the optimum restart condition~$(p-1)Q_X(p)+p(1-p)Q'_X(p)-\Q_X(p)+\Q_X(0)=b/a$. If $a\gg b$, the optimal restart time approaches $Q(0)=b$, i.e., the algorithm's behavior before $b$ dominates all subsequent steps.
%
%
It is also remarkable 
that some distributions which show suboptimal behavior without restarts yield low runtimes when an optimal restart strategy is applied.